\documentclass[letterpaper, 10 pt, conference]{ieeeconf}  % Comment this line out if you need a4paper

\IEEEoverridecommandlockouts                              % This command is only needed if 
\usepackage{times} % assumes new font selection scheme installed
\usepackage{amsmath} % assumes amsmath package installed
\usepackage{amssymb}  % assumes amsmath package installed
\usepackage{graphicx}
\usepackage{dsfont}
\usepackage{color}
\usepackage{multicol}
\usepackage{url}
\usepackage{subcaption}
\usepackage{algorithm}
\usepackage{algpseudocode}
\usepackage{algorithmicx}
\usepackage{cite}
\usepackage{hyperref}
\usepackage{flushend}

\usepackage{amsthm}
\newtheorem{theorem}{Theorem}

\newtheorem{proposition}{Proposition}
\newtheorem{remark}{Remark}
\newtheorem{lemma}{Lemma}
\newtheorem{problem}{Problem}

\theoremstyle{definition}
\newtheorem{assumption}{Assumption}

\usepackage{comment}
\newcommand{\Figref}[1]{Figure\,\ref{#1}}
\newcommand{\figref}[1]{Fig.\,\ref{#1}}
\newcommand{\secref}[1]{Sec.\,\ref{#1}}

\newcommand{\mE}{{\mathbb{E}}}
\newcommand{\mP}{{\mathbb{P}}}
\newcommand{\mI}{{\mathds{1}}}
\newcommand{\R}{\mathbb{R}}
\newcommand{\mZ}{\mathbb{Z}}
\newcommand{\SProb}{\Psi}
\newcommand{\SSet}{\mathcal{C}}

\DeclareMathOperator*{\argsup}{arg\,sup}

\title{\LARGE \bf
Physics-informed RL %Reinforcement Learning 
for Maximal Safety Probability Estimation %*
}

\author{Hikaru Hoshino$^{1}$ and Yorie Nakahira$^{2}$% <-this % stops a space
\thanks{*This work was supported in part by Grant-in-Aid for Scientific Research (KAKENHI) from the Japan Society for Promotion of Science (\#23K13354), in part by the PRESTO Grant Number JPMJPR2136 from Japan Science and Technology agency, in part by the Department of the Navy, Office of Naval Research, grant number N00014-23-1-2252, and in part by Mobility21 National University Transportation Center, which is sponsored by the US Department of Transportation. 
Any opinions, findings, and conclusions or recommendations expressed in this material are those of the authors and do not necessarily reflect the views of the Office of Naval Research.
}% <-this % stops a space
\thanks{$^{1}$Hikaru Hoshino is with the Department of Electrical Materials and Engineering, University of Hyoto, 
        2167 Shosha, Himeji, Hyogo, 671-2280 Japan
        {\tt\small hoshino@eng.u-hyogo.ac.jp}}%
\thanks{$^{2}$Yorie Nakahira is with the Department of Electrical and Computer Engineering, Carnegie  Mellon  University, 
        4815 Frew St, Pittsburgh, PA 15213, USA
        {\tt\small ynakahira@andrew.cmu.edu}}%
}

\renewcommand{\baselinestretch}{0.97}

\begin{document}

\maketitle
\thispagestyle{empty}
\pagestyle{empty}

%%%%%%%%%%%%%%%%%%%%%%%%%%%%%%%%%%%%%%%%%%%%%%%%%%%%%%%%%%%%%%%%%%%%%%%%%%%%%%%%
\begin{abstract}

Accurate risk quantification and reachability analysis are crucial for safe control and learning, but sampling from rare events, risky states, or long-term trajectories can be prohibitively costly. Motivated by this, we study how to estimate the long-term safety probability of maximally safe actions without sufficient coverage of samples from risky states and long-term trajectories. The use of maximal safety probability in control and learning is expected to avoid conservative behaviors due to over-approximation of risk. Here, we first show that long-term safety probability, which is multiplicative in time, can be converted into additive costs and be solved using standard reinforcement learning methods. We then derive this probability as solutions of partial differential equations (PDEs) and propose Physics-Informed Reinforcement Learning (PIRL) algorithm. The proposed method can learn using sparse rewards because the physics constraints help propagate risk information through neighbors. This suggests that, for the purpose of extracting more information for efficient learning, physics constraints can serve as an alternative to reward shaping. The proposed method can also estimate long-term risk using short-term samples and deduce the risk of unsampled states. This feature is in stark contrast with the unconstrained deep RL that demands sufficient data coverage. These merits of the proposed method are demonstrated in numerical simulation.

\end{abstract}

%%%%%%%%%%%%%%%%%%%%%%%%%%%%%%%%%%%%%%%%%%%%%%%%%%%%%%
\section{Introduction}

Risk quantification and reachability analysis are crucial for safety-critical autonomous control systems. For example, these techniques are widely used in stochastic safe control~\cite{Chapman2019,Wang2022}, safe exploration~\cite{Berkenkamp2019,Kim2021}, and safe reinforcement learning~\cite{Srinivasan2020,Qin2021,Chen2023,Thananjeyan2020}. 
%Overestimation of risk based on suboptimal actions can excessively limit actions and exploration. 
However, it is challenging to accurately quantify long-term risks and find maximally safe control policies for complex nonlinear systems.  
There are stringent trade-offs between accuracy, time horizon, sample complexity, and computation. Such tradeoffs are particularly stringent when the risk is associated with rare events and the dimensions of the systems are high~\cite{Zhang2021}. %\cite{Wang2004,Janssen2013,Zhang2021}. 
In addition, unsafe events, risky states, and long-term trajectories can be prohibitively costly to sample from physical systems. 
Motivated by these challenges, this paper proposes an efficient Physics-Informed Reinforcement Learning (PIRL) that can estimate long-term maximal safety probabilities with \textit{short-term} data that do \textit{not} contain many unsafe events.  

%The purpose of this paper is to propose a physics-informed RL framework to efficiently estimate safe/risk probabilities of maximally safe actions while avoiding undesirable overestimation that excessively constrains states and actions. 

%
% Safe control for stochastic systems is important yet a key challenge for deploying autonomous systems in the real world. 
% In the past decades, many stochastic control methods have been proposed to ensure safety of systems with noises and uncertainties, including stochastic reachabilities \cite{Chapman2019}, 
% barrier functions \cite{Santoyo2021,Wang2022},
% chance-constrained predictive control \cite{Nakka2021}, and safe Reinforcement Learning (RL) \cite{Garcia2015,Gu2022}. 
% Despite the huge amount of stochastic safe control methods, many of them rely on accurate estimates of long-term risk probabilities or their gradients to guarantee long-term safety. 
%To obtain such accurate estimates, a physics-informed learning framework is proposed in \cite{Wang2023}, where a technique of Physics-Informed Neural Networks (PINN) \cite{Raissi2019} is used for risk probability estimation. 
%While a great performance has been shown regarding sample efficiency and generalization to unseen regions in the state space, the method in  \cite{Wang2023} requires a prescribed nominal controller and the resultant risk probability depends on the choice of the nominal controller. 
%The purpose of this paper is to develop a physics-informed RL framework to concurrently learn a maximally safe controller and the associated risk probability. 
%

%\subsection{Related work}

Many learning-based techniques were developed to quantify various forms of risk. For deterministic systems (worst-case framework), RL techniques were adapted for reachability analysis~\cite{Fisac2019}. For stochastic systems, policy gradient approaches were used to minimize CVaR and coherent risk measures~\cite{Tamar2015}. 
Deep Q-learning was used to learn the probabilities of constraint violations of time horizon one (at each time), which are then used to constrain learning and exploration~\cite{Thananjeyan2020,Srinivasan2020}.
Estimation of long-term probabilities under maximally safe actions is an optimization problem with multiplicative costs over time whose optimality conditions were characterized~\cite{Abate2008}. 
However, solving such optimization problems is not trivial, particularly for high-dimensional systems. Although techniques such as taking logarithms are often used to convert multiplicative costs into summations in practice, such techniques cannot be used directly in this setting (Remark\,\ref{remark:multiplicative} for details). Here, we show that long-term safety probabilities (in the form of multiplicative costs of expected index functions) are transferable to additive costs, for which many RL methods can be used. 

Although RL has the potential to offer scalable risk quantification techniques, one may not know how accurate the converged solutions and generalization to states or time horizons whose samples are unavailable. This is problematic if the quantified risk is to be used in safety-critical systems, because the safety of subsequent decision-making techniques depends on accurate risk quantification. To tackle these challenges, we propose to leverage Physics-Informed Neural Networks (PINN)~\cite{Raissi2019}.  PINN has a demonstrated potential in generalization due to the use of physics constraints~\cite{Cai2021,Cuomo2022}. 
PINN-based approach has been used to quantify safety probabilities of a given controller with provable generalization~\cite{Wang2023}.
%It uses a PDE characterization of the safety probability derived in \cite{chern21} as a variant of the Feynman-Kac representation of stochastic processes. 
%On the other hand, the framework of probabilistic reachability \cite{Abate2008} and reach-avoid problems \cite{Summers2010} has characterized the maximum safety probability. 
%However, the optimality conditions derived in their work do not take a partial differential equation (PDE) form and cannot be directly used for imposing a physics constraint. Here, we further derive a PDE characterizing the safety probability under the maximally safe controller and integrate it into a PIRL framework. 
Here, we derive a PDE characterizing the safety probability %under the maximally safe controller based on the framework of reach-avoid problems \cite{MohajerinEsfahani2016} \todo{do we have 'reach' part in our problem?} 
and integrate it into a PIRL framework. 
%based on the Hamilton-Jacobi-Bellman (HJB) theory for exit-time problems developed in \cite{MohajerinEsfahani2016}, and integrate it into a PIRL framework.  

%In this paper, by 1) introducing an appropriate augmentation of the state space and 2) using the idea in \cite{Summers2010} of representing the cost in the form of sum of multiplicative costs, we show that the original multiplicative formulation of the safety probability can be  naturally transformed to a standard additive cost.  

Due to the integration of RL and PINN, the proposed framework has the following advantages. 
\begin{itemize}
    \item Expansion of feasible regions: By exploring a maximally safe controller, the set of state spaces with tolerable risks is expanded. When the maximal safety probability is used to constrain action and exploration, the system is expected to be less conservative (see \figref{fig:safety_prob}).
    %\item Faster learning: The proposed method has improved sample efficiency due to the joint use of prior knowledge and sampling (see Fig.\,2). 
    \item Learning from sparse rewards in space and time: The proposed method can learn from binary rewards that are also sparse in time and achieve objectives similar to reward shaping (see \figref{fig:reward_shaping}). This is achieved by leveraging physics constraints to extract and propagate information from neighbors and boundaries. 
    \item Generalization to longer-horizon and unsampled risky states: The proposed method can estimate long-term safety probability using short-term samples and achieve comparable learning effect using reduced number of unsafe events (see \figref{fig:generalization}). This feature is beneficial when samples from long-term trajectories are unavailable or when risky states are costly to sample. 
\end{itemize}
The proposed method is built on Deep Q-Network (DQN) algorithm \cite{Mnih2013}, but the framework is generalizable to other deep RL techniques. 
While several PIRL frameworks have been proposed (see \cite{Banerjee2023} for a review), to the best of our knowledge, this work is the first to combine an RL problem with PINN for the purpose of estimating maximal safety probabilities and the corresponding policies.

\subsection{Notation}

% In this section, after introducing preliminary notations in \secref{sec:notation}, we introduce the system description and safety specifications in \secref{sec:system}, and present the problem statement in \secref{sec:problem}. 

% \subsection{Preliminary Notations} \label{sec:notation}

Let $\mathbb{R}$ and $\mathbb{R}_+$ be the set of real numbers and the set of nonnegative real numbers, respectively. 
Let $\mZ$ and $\mZ_+$ be the set of integers and the set of non-negative integers. 
For a set $A$, $A^\mathrm{c}$ stands for the complement of $A$, and $\partial A$ for the boundary of $A$. 
Let $\lfloor x \rfloor \in \mZ$ be the greatest integer less than or equal to $x\in\R$.  
Let $\mathds{1}[\mathcal{E}]$ be an indicator function, which takes 1 when the condition $\mathcal{E}$ holds and otherwise 0.
Let $\mP[ \mathcal{E} | X_0 = x ]$ represents the probability that the condition $\mathcal{E}$ holds involving a stochastic process $X = \{ X_t \}_{t\in\R_+}$ conditioned on $X_0 = x$.  
Given random variables $X$ and $Y$, let $\mathbb{E}[X]$ be the expectation of $X$, and $\mathbb{E}[X|Y=y]$ be the conditional expectation of $X$ given $Y=y$. We use upper-case letters (\emph{e.g.}, $Y$) to denote random variables and lower-case letters (\emph{e.g.}, $y$) to denote their specific realizations. 
For a scalar function $\phi$, % with an argument $x \in\R^n$, 
$\partial_x \phi$ stands for the gradient of $\phi$ with respect to $x$, and $\partial_{x}^2 \phi$ for the Hessian matrix of $\phi$. 
Let $\mathrm{tr}(M)$ be the trace of the matrix $M$. 

%%%%%%%%%%%%%%%%%%%%%%%%%%%%%%%%
\section{Problem Statement}

We consider a control system with stochastic noise of $w$-dimensional Brownian motion $\{ W_t \}_{t \in \R_+}$ starting from $W_0 = 0$. 
The system state $X_t \in \mathbb{X} \subset \mathbb{R}^n$ evolves according to the following stochastic differential equation (SDE): 
\begin{align}
 \mathrm{d}X_t = f(X_t, U_t) \mathrm{d}t + \sigma(X_t, U_t) \mathrm{d}W_t,  \label{eq:sde}   
\end{align}
% \begin{align}
%  \color{red}
%  \mathrm{d}X_t = f(X_t, u( X_{\delta(t)}, t)) \mathrm{d}t + \sigma(X_t) \mathrm{d}W_t,  \label{eq:sde}   
% \end{align}
% {\color{red}where the control $u(X_{\delta(t)}, t)$ with $\delta(t):= \lfloor t/\Delta t \rfloor \Delta t$ is based on state observations at discrete times, $0, \,\Delta t, \,2\Delta t, \dots $, with the time step $\Delta t$.}
where $U_t \in \mathbb{U} \subset \mathbb{R}^m$ is the control input. 
Throughout this paper, we assume sufficient regularity in the coefficients of the system \eqref{eq:sde}.
That is, the functions $f$ and $\sigma$ are chosen in a way such that the SDE \eqref{eq:sde} admits a unique strong solution (see, e.g., Section IV.2 of \cite{Fleming06}). 
The size of $\sigma(X_t)$ is determined from the uncertainties in the disturbance, unmodeled dynamics, and prediction errors of the environmental variables.

For numerical approximations of the solutions of the SDE and optimal control problems, we consider a discretization with respect to time with a constant step size $\Delta t$ under piecewise constant control processes. 
For $0=t_0 < t_1 < \dots < t_k < \dots$, where $t_k := k\Delta t$, $k \in \mZ_+$, by defining the discrete-time state $X_k := X_{t_k}$ with an abuse of notation, the discretized system can be given as
%SDE \eqref{eq:sde} can be discretized  %to obtain the dynamics of $Z_k := X_{t_k}$ 
% \begin{align}
%   Z_0 = X_0, \quad Z_{k+1} = F(Z_k, A_k, \Delta W_k), \quad \forall k \in \mZ_+, \label{eq:descrite_system}
% \end{align}
\begin{align}
  X_{k+1} = F^u(X_k, \Delta W_k), \label{eq:descrite_system}
\end{align}
where $\Delta W_k := \{ W_t \}_{t \in [t_k, t_{k+1})}$, and $F^{u}$ stands for the state transition map derived from \eqref{eq:sde} under a Markov control policy $u : [0, \infty) \times \mathbb{X} \to \mathbb{U}$. 
From an optimal control perspective, using a Markov policy is not restrictive when the value function has a sufficient smoothness under several technical conditions (see \cite[Theorem IV.4.4]{Fleming06} and Assumption\,\ref{assumption1} below). 
Note that using a piece-wise constant control process with a Markov policy $u$ implies that the control process is given as $U_t = u(\delta(t), X_{\delta(t)})$, for $t\in \R_+$,  where $\delta(t) := \lfloor t/\Delta t \rfloor \Delta t$, and the discretized system \eqref{eq:descrite_system} has the Markov property at the discrete times \cite{Mao2013}.

Safety of the system can be defined by using a safe set $\SSet \subset \mathbb{X}$.  
%For the discretized system \eqref{eq:descrite_system} and for a given sequence of control actions $\{ a_k \}_{k=0}^\infty$, 
For the discretized system \eqref{eq:descrite_system} and for a given control policy $u$, 
the safety probability $\SProb^{u}$ of %the system \eqref{eq:descrite_system} 
the initial state $X_0 = x$ for the outlook horizon $\tau \in \R$ %, with some integer $N \in \mZ_+$, 
can be characterized as %is defined as 
the probability that the state $X_k$ stays within the safe set $\SSet$ for %during the interval 
$k \in \mathcal{N}_\tau := \{0, \dots, N(\tau)\}$, where $N(\tau) := \lfloor \tau / \Delta t \rfloor$, \emph{i.e.}, 
% \begin{align}
%     \SProb(T, x; \{ a_k \}_{k=0}^\infty) := & \mP[ Z_k \in \SSet,\, \forall k \in \mathcal{N} ~|~ \notag \\
%      & \quad Z_0 = x,\, A_k = a_k,  \forall k \in \mZ_+ ]. %\mathcal{N}  ]. 
% \end{align}
\begin{align}
\SProb^u(\tau, x)
:= \mP[ X_k \in \SSet,\, \forall k \in \mathcal{N}_\tau ~|~ X_0 = x, u ].  
\end{align}
%Note that the control action $a_k$ for $k \ge N$ does not affect the value of the safety probability, but we include it for the consistency of the notation in the subsequent  discussion. 
Then, the objective of this paper can be described as follows.  
\begin{problem} \label{problem1}
   Consider the system \eqref{eq:descrite_system} starting from an initial state $x \in \SSet$. Then, estimate the maximal safety probability defined as 
   \begin{align}
     \Psi^\ast(\tau,x) := \sup_{ u \in \mathcal{U} } \Psi^u(\tau,x),
   \end{align}
   where $\mathcal{U}$ is the class of bounded and Borel measurable Markov control policies.    
%   satisfies a sufficient regularity condition such that $\Psi$  converges to $\Psi^u_\mathrm{c}(T,x):= \mP[X_t\in \SSet, \forall t \in[0,T]\,|\, X_0=x, u ]$ as $\Delta t \to 0$. 
\end{problem}

For the results stated in the next section, we assume the following technical conditions:
\begin{assumption} \label{assumption1}
 We stipulate that 
 \begin{enumerate}
   \item[(a)] $\mathbb{U}$ is compact. 
   \item[(b)] $f$, $\sigma$ and their first and second partial derivatives with respect to the state are continuous.  
   \item[(c)] $\sigma(x, u)$ is an $n \times n$ matrix, such that for all $(x, u) \in \mathbb{X}\times \mathbb{U}$ and $\xi \in \R^n$, $\sum_{i,j=1}^n \sigma_{ij}(x,u) \xi_i \xi_j \ge \gamma |\xi|^2$, where $\gamma >0$. 
   \item[(d)] $\SSet^\mathrm{c}$ is a bounded closed subset of $\mathbb{X}$ with $\partial \SSet$, a three-times continuously differentiable manifold.
   \item[(e)] $\Psi^u(\tau, x)$ converges to $\Psi^u_\mathrm{c}(\tau,x):= \mP[X_t\in \SSet, \forall t \in[0,\tau]\,|\, X_0=x, u ]$ as $\Delta t \to 0$. \label{as:safe_prob_limit}
 \end{enumerate}
\end{assumption}
The assumptions (a) to (d) are used for assuring the smoothness of the value function discussed in \secref{sec:pde}. 
The assumption (e) is needed to ensure the consistency between the safety probability in the discrete time and the PDE condition in the continuous time, and similar conditions are achieved in \cite{Mao2013,Bayraktar2023}.

%%%%%%%%%%%%%%%%%%%%%%%%%%%%%%%%%%%%%%%%
\section{Proposed Framework}

Here we present a physics-informed RL framework for  safety probability estimation. 
For this, a problem formulation with additive cost is presented in \secref{sec:rl_probem}, and a PDE characterization for the safety probability is derived in \secref{sec:pde}. 
The proposed framework is presented in  \secref{sec:framework}.

\subsection{Problem Formulation with Additive Cost} \label{sec:rl_probem}

Problem\,\ref{problem1} can be regarded as a stochastic optimal control problem with a multiplicative cost %\eqref{eq:safe_prob} 
to be maximized, because the objective function $\SProb^u$ is naively written as follows: 
% \begin{align}
%    \SProb(T, x; \{ a_k \}_{k=0}^\infty)   
%    %%%%%%%%%%%%%%%%%
%    = &\mathbb{E}[ \mathds{1}\left[ Z_k \in \SSet, \, \forall k \in \mathcal{N} \, \right]  \,|\, \notag \\&\quad 
%    Z_0 = x , A_k = a_k, \forall k \in \mZ_+ 
%    ] \notag \\
%    %%%%%%%%%%%%%%%%%%
%     = &\mathbb{E}\Biggl[ \prod_{k=0}^N \mathds{1}[ Z_k \in \mathcal{C} ] \,\Bigg|\, \notag \\
%     & \quad  Z_0 = x,  A_k = a_k, \forall k \in \mZ_+ 
%     \biggr], %\notag \\
%     %= &\mathbb{E}\left[ P_{0, N}  \,\middle|\,  Z_0 = x, A_k = a_k, \forall k \in \mathcal{N} \right] 
%     \label{eq:safe_prob}
% \end{align}
\begin{align}
   \SProb^u(\tau, x)   
   %%%%%%%%%%%%%%%%%
   = &\mathbb{E}[ \mathds{1}\left[ X_k \in \SSet, \, \forall k \in \mathcal{N}_\tau \, \right]  \,|\, X_0 = x , u ] \notag \\
   %%%%%%%%%%%%%%%%%%
    = &\mathbb{E}\Biggl[ \prod_{k=0}^{N(\tau)} \mathds{1}[ X_k \in \SSet ] \,\Bigg|\, X_0 = x, u \Biggr], 
    \label{eq:safe_prob}
\end{align}

\begin{remark} \label{remark:multiplicative}
To convert a multiplicative cost into an additive cost, there are two typical ways taken in RL problem formulations. 
One is to use a log scale translation of the return. %to rewrite a multiplicative cost into an additive form, 
However, this approach fails in the case of the safety probability. This is because each term is conditioned on the previous steps, and thus the reward at the time step $k$ can not be represented as a function of the state $X_k$ as follows:  
\begin{align}
   \log \SProb^u (\tau, x) =&  \log \mP[ \cap_{k=0}^{N(\tau)} \mathcal{E}_k | X_0 =x, u ] \notag \\
   = & \log \mP[\mathcal{E}_0| X_0 =x, u]\, \mP[\mathcal{E}_1|\mathcal{E}_0, X_0 =x, u] \notag \\
     & \cdot \mP[\mathcal{E}_2|\mathcal{E}_{1}, \mathcal{E}_0, X_0 =x, u] \cdots \notag \\
   = & \log \prod_{k=0}^{N(\tau)} \mP[\mathcal{E}_k|\mathcal{E}_{k-1},\dots, \mathcal{E}_0, X_0 =x, u] \notag \\
   = &  \sum_{k=0}^{N(\tau)} \log \mP[\mathcal{E}_k|\mathcal{E}_{k-1}, \dots, \mathcal{E}_0, X_0 =x, u],  
\end{align}
where $\mathcal{E}_k$ represents the condition that $X_k \in \SSet$.
The second approach is to augment the state space by considering a sequence of observations as a state, i.e., $\{ x_0, x_1, \dots, x_k\}$. 
In this paper, we will consider an augmented state that is only one-dimension higher than the original state, which significantly reduces the dimension of the state space. \end{remark}

In this paper, by 1) introducing an appropriate augmented system, and 2) using the idea in \cite{Summers2010} of representing the cost in a form of sum of multiplicative costs, 
we show that the above multiplicative cost %stochastic control problem 
can be naturally transformed to an additive cost.
For this, we consider a variable $H_k$ that represents the remaining time before the outlook horizon $\tau$ is reached, \emph{i.e.}, 
\begin{align}
    H_0 = \tau, \quad H_{k+1} = H_{k} - \Delta t.  
\end{align}
Then, let us consider the augmented state space $\mathcal{S} := \R \times \mathbb{X} \subset \R^{n+1}$ and the augmented state $S_k \in \mathcal{S}$, where we denote the first element of $S_k$ by $\tilde{H}_k$ and the other elements by $\tilde{X}_k$, \emph{i.e.}, 
\begin{align}
 S_k =[\tilde{H}_k, \tilde{X}_k^\top]^\top,  
\end{align}
where we use the tilde notation to distinguish between the original
dynamics \eqref{eq:descrite_system} and those for the additive cost representation introduced below. 
For the state $S_k$, consider the stochastic dynamics starting from the initial state
\begin{align}
 S_0 = s := [\tau, x^\top]^\top \in \mathcal{S}
\end{align}
with $\tau \in \R$ given as follows: 
for $\forall k \in \mZ_+$, 
\begin{align}
  S_{k+1} = 
  \begin{cases}
      \tilde{F}^u(S_k, \Delta W_k),  & S_k \notin \mathcal{S}_\mathrm{abs},   \\
      S_k,  & S_k \in \mathcal{S}_\mathrm{abs}, 
  \end{cases}
  %\quad \forall k \in \mZ_+
 \label{eq:augmented_dynamics}
 \end{align}
with the function $\tilde{F}^u$ given by 
\begin{align}
  \tilde{F}^u(S_k, \Delta W_k) := 
 \begin{bmatrix} 
    \tilde{H}_k - \Delta t \\
    F^u(\tilde{X}_k, \Delta W_k) 
 \end{bmatrix},   
\end{align}
and the set of absorbing states $\mathcal{S}_\mathrm{abs}$ given by 
\begin{align}
 \mathcal{S}_\mathrm{abs} := \{ [\tilde{\tau}, \tilde{x}^\top]^\top \in \mathcal{S} ~|~ \tilde{\tau} < 0 \, \lor \, \tilde{x} \in \SSet^\mathrm{c}   \}.
 \label{eq:terminal_states}
\end{align}
The notion of absorbing state is commonly used in RL literature \cite{sutton18}, and %when $\tau$ is chosen as $\tau = T$, 
we have $S_k = [\tilde{H}_k, \tilde{X}_k^\top]^\top = [H_k, X_k^\top]^\top$ for the states satisfying $S_k \notin \mathcal{S}_\mathrm{abs}$, but not for $S_k \in \mathcal{S}_\mathrm{abs}$ where the state $S_k$ transitions to itself.

Then, the following proposition states that the multiplicative cost representation \eqref{eq:safe_prob} can be transformed to an additive cost by using the augmented dynamics \eqref{eq:augmented_dynamics}. 

\begin{proposition} \label{prop:RL}
Consider %an RL problem with the state $S_k$, $\forall k \in \mZ_+$, with the underlying dynamics given by
the system \eqref{eq:augmented_dynamics} starting from an initial state $s = [\tau, x^\top]^\top \in \mathcal{S}$ %for $S_k \notin \mathcal{S}_\mathrm{abs}$ with the absorbing states given in \eqref{eq:terminal_states} 
and the reward function $r: \mathcal{S} \to \R$ given by
\begin{align} \label{eq:reward}
  r(S_k) := 
      \mI[ \tilde{H}_k \in \mathcal{G} ]\, \mI[ S_k \notin \mathcal{S}_\mathrm{abs}] 
\end{align}
with $\mathcal{G} := [0, \Delta t)$. 
%Then, for a given control sequence $\{ a_k \}_{k=0}^\infty$, the objective function $J$ given by  
Then, for a given control policy $u$, the value function $v^{u}$ defined by  
\begin{align}
  &v^{u}(s) := \mathbb{E} \left[ \sum_{k=0}^\infty r( S_k ) ~\middle | ~ S_0= s, u  \right]  \label{eq:safe_ocp_cost}
\end{align}
% \begin{align} \label{eq:value_function}
%     J(s) := \mathbb{E}_K\left[ \sum_{k=0}^\infty R(S_k) ~\middle|~ S_0 = s \right]
% \end{align}    
takes a value in $[0,1]$ and is equivalent to the safe probability $\SProb^u(\tau, x)$, i.e., 
\begin{align}
 v^{u}(s) = \SProb^{u}(\tau, x).     
\end{align}
\end{proposition}

\begin{proof}
    See Appendix\,A. 
\end{proof}

\begin{comment}
\begin{remark}
 The proof of Proposition\,1 is based on the idea used in \cite{Summers2010} to rewrite a multiplicative cost into a form of sum of multiplicative costs. 
 In \cite{Summers2010}, this technique was used for deriving the optimality condition of an optimal control problem with a multiplicative cost. 
 Here, we show that with the augmentation of the state to include $\tilde{\tau}_k$ and considering the set of absorbing states $\mathcal{S}_\mathrm{abs}$, the sum of multiplicative costs can be simplified to a standard additive cost. 
\end{remark}
%
From Proposition\,1, Problem\,1 can be rewritten as follows:
\begin{problem}
    Consider the system \eqref{eq:augmented_dynamics} starting from $s = [\tau, x^\top]^\top \in \mathcal{S}$. 
    Then, solve the optimal control problem with the additive cost $J$ in \eqref{eq:safe_ocp_cost} 
    to find $\tilde{\Psi}^\ast$  given by 
    \begin{align}
       \Psi^\ast(\tau, x) = \sup_{ {\color{red}u \in \mathcal{U}} }  J^{\color{red}u}(s)  
    \end{align}
    {\color{blue}
    where $a_k \in \mathcal{U}$ for $k \in \mZ_+$, 
   and any admissible control sequence satisfies a sufficient regularity condition such that $J$ converges to $\Psi_\mathrm{c}(T,x; \{u_t\}_{t\in \R_+})$ as $\Delta t \to 0$ with $\tau = T$.} 
\end{problem}
\end{comment}

Since the reward function $r$ contains the term $\mI[S_k \notin \mathcal{S}_\mathrm{abs}]$, the reward is always zero for $S_k \in \mathcal{S}_\mathrm{abs}$. Thus, the value function $v^u$ can also be written as 
\begin{align}
     v^{u}(s) = \mathbb{E} \left[ \sum_{k=0}^{N_\mathrm{f}} r( S_k ) ~\middle | ~ S_0= s, u  \right], \label{eq:safe_ocp_cost}
\end{align}
where $N_\mathrm{f}$ is the first entry time to $\mathcal{S}_\mathrm{abs}$ given by 
\begin{align}
    N_\mathrm{f} := \inf \left\{ j \in \mZ_+ \,\middle|\,  S_j \in \mathcal{S}_\mathrm{abs} \right\}. 
\end{align}
Thus, we can consider an episodic RL problem by treating $\mathcal{S}_\mathrm{abs}$ as the terminal states. 
The action-value function $q^u(s, a)$, defined as the value of taking an action $a \in \mathbb{U}$ in state $s$ and thereafter following the policy $u$, is given by 
\begin{align}
 q^u(s, a) := \mathbb{E}\left[ \sum_{k=0}^{N_\mathrm{f}} r(S_k) ~\middle|~ S_0 = s, U_0 = a, u  \right].
\end{align}
The objective of RL is to find the optimal action-value function defined as 
\begin{align}
  q^\ast(s, a) := \sup_{ u \in \mathcal{U} } q^u(s, a).
\end{align}

% {\color{blue}
% Furthermore, it is proved that the value function of the RL problem is equivalent to the safety probability $\SProb_{K}(T, x)$, meaning that the learned value function can be implemented as a function approximator of the safety probability.  
% }

%%%%%%%%%%%%%%%%%%%%%
\subsection{PDE Characterization of Safety Probability} \label{sec:pde}

To implement the technique of PINN, a PDE condition is introduced in this subsection. 
This is achieved based on the Hamilton-Jacobi-Bellman (HJB) theory of stochastic optimal control for a class of reach-avoid problems~\cite{MohajerinEsfahani2016}. 
The safety problem can be regarded as a special case of reach-avoid problems, which determines whether there exists a control policy such that the process $X$ reaches a target set $A$ prior to entering an unsafe set $B$.  
In \cite{MohajerinEsfahani2016}, for the continuous-time setting of the SDE \eqref{eq:sde}, an exit-time problem is considered to characterize the function given by 
\begin{align}
 V^{\bar{u}}(t_\mathrm{s},x) := \mE [\mI [X_{T_\mathrm{e}}^{t_\mathrm{s},x; \bar{u} } \in A ]], ~ T_\mathrm{e} := \min( T_B, t_\mathrm{f} ),
 \label{eq:ocp_v}
\end{align}
where the process $\{ X_t^{t_\mathrm{s},x;\bar{u}} \}_{t \in [t_\mathrm{s}, t_\mathrm{f}]}$ represents the unique strong solution of \eqref{eq:sde} for the time interval of $[t_\mathrm{s}, t_\mathrm{f}]$ starting from the state $x$ under the control process $\bar{u}$, which belongs to the set $\mathcal{U}_{t_\mathrm{s}}$ of progressively measurable maps into $\mathbb{U}$. 
The random variable $T_B$ stands for the first entry time to $B$.
By taking $A:= \mathcal{C}$ and $B:= \mathcal{C}^\mathrm{c}$, the function $V^{\bar{u}}(t_\mathrm{s}, x)$ can be rewritten as  
\begin{align}
 V^{\bar{u}}(t_\mathrm{s},x) & = \mE[\mI[X_{T_\mathrm{e}}^{t_s,x;\bar{u}}\in \SSet]] \\
  & = \mE[\mI[X_{t}^{t_s,x;\bar{u}}\in \SSet, \, \forall t \in [t_\mathrm{s}, t_\mathrm{f}]]], 
 \end{align}
where the second equality holds because $\mI[X_{T_\mathrm{e}}^{t_s,x;\bar{u}}\in \SSet] =1$ if and only if the state $X_t$ stays in $\SSet=B^\mathrm{c}$ for $t\in[t_\mathrm{s}, t_\mathrm{f}]$ (see \cite[Proposition\,3.3]{MohajerinEsfahani2016} for a precise discussion). 
Thus, with Assumption\,\ref{assumption1}(d),
we have 
\begin{align}
  V^{\bar{u}}(t_\mathrm{f}-\tau, x) = \lim_{\Delta t \to 0} \Psi^u(\tau, x) = \lim_{\Delta t \to 0} v^u(s), 
\end{align}
when we choose the control process $\bar{u}$ such that it determines the control input as
$U_t = u(t, X_t)$. %, for $t \in \R_+$. 

In \cite{MohajerinEsfahani2016}, the optimal value function $V^\ast(t_\mathrm{s}, x) := \sup_{\bar{u} \in \mathcal{U}_{t_\mathrm{s}} } V^{\bar{u}}(t_\mathrm{s}, x)$, is characterized as a solution of an HJB equation.
However, it does not admit a classical solution due to the discontinuity of the payoff function given by the indicator function. 
Instead, $V^\ast$ becomes a discontinuous viscosity solution of a PDE under mild technical conditions \cite[Theorem 4.7]{MohajerinEsfahani2016}. 
Furthermore, to allow the use of numerical solution techniques mainly developed for continuous or smooth solutions, it is shown in \cite{MohajerinEsfahani2016} that one can construct a slightly conservative but arbitrarily precise way of characterizing the original solution by considering a set $A_\epsilon$ smaller than $A$, where $A_\epsilon := \{ x \in A \,|\, \mathrm{dist}(x, A^\mathrm{c}) \ge \epsilon \}$, with $\text{dist}(x, A) := \inf_{y \in A} \| x - y \|$.  
Following \cite{MohajerinEsfahani2016}, to derive a PDE condition to implement PINN, we consider the following function: 
\begin{align}
 q^u_\epsilon(s, a) := \mathbb{E}\left[ \sum_{k=0}^{N_\mathrm{f}} r_\epsilon(S_k) ~\middle|~ S_0 = s, U_0 = a, u  \right], 
\end{align}
where  the function $r_\epsilon$ is given by 
\begin{align}
  r_\epsilon(S_k) := 
      \mI[ \tilde{H}_k \in \mathcal{G} ]\, \mI[ S_k \notin \mathcal{S}_\mathrm{abs}]\, l_\epsilon( \tilde{X}_k ),
\end{align}
with $\SSet_\epsilon := \{ x \in \SSet \,|\, \mathrm{dist}(x, \SSet^\mathrm{c}) \ge \epsilon \}$ and 
\begin{align}
 l_\epsilon (x) :=  \max \left\{  1- \dfrac{\text{dist}(x, \SSet_\epsilon)}{\epsilon}, \, 0\right\}. 
\end{align}

\begin{theorem} \label{thm:hjb}
Consider the system \eqref{eq:augmented_dynamics} derived from the SDE \eqref{eq:sde} and suppose that Assumption\,\ref{assumption1} holds. 
Then, for all $s\in \mathcal{S}$ and $a \in \mathbb{U}$,  $q^\ast(s, a) = \lim_{\epsilon \to 0} q^\ast_\epsilon(s, a)$, where $q^\ast_\epsilon(s, a) := \sup_{ u \in \mathcal{U} } q^u_\epsilon(s, a)$.
Furthermore, the function $q^\ast_\epsilon(s, a)$ is the continuous viscosity solution of the following partial differential equation in the limit of $\Delta t \to 0$:  
for $s \in (0, \infty) \times \SSet$, 
\begin{align} 
      & %\dfrac{ \partial q^\ast_\epsilon(s, a^\ast)}{ \partial s }
      \partial_s  q^\ast_\epsilon(s, a^\ast) 
      %\bigg|_{s=s, a=a^\ast} 
      \tilde{f}(s, a^\ast) 
      \notag \\ & \hspace{5mm} 
      +\dfrac{1}{2} \mathrm{tr} \left[ \tilde{\sigma}(s,a^\ast)\tilde{\sigma}(s,a^\ast)^\top  \partial_s^2 q^\ast_\epsilon(s, a^\ast)  %\bigg|_{s=s, a=a^\ast}  
      \right] = 0,  \label{eq:pde}
\end{align} 
where the function $\tilde{f}$ and $\tilde{\sigma}$ are given by  
\begin{align}
    \tilde{f}(s, a) := 
    \begin{bmatrix}
       -1 \\   f(x, a)
    \end{bmatrix},
    \quad
    \tilde{\sigma}(s, a) := 
    \begin{bmatrix}
        0 \\ \sigma(x, a)        
    \end{bmatrix}, 
\end{align}
%where $D(x) := \sigma(x) \sigma(x)^\top$ 
and $a^\ast := \argsup_{a \in \mathbb{U}} q^\ast(s, a)$. 
The boundary conditions are given by 
\begin{align}
  & q^\ast_\epsilon([0, x^\top]^\top, a^\ast) = l_\epsilon( x ), \quad  \forall x \in \mathbb{X}, \label{eq:boundary_tau_zero} \\
    & q^\ast_\epsilon([\tau, x^\top]^\top, a^\ast) = 0, \quad   \forall \tau \in \R, ~ \forall x \in \partial\SSet.    \label{eq:boundary_unsafe}
\end{align}
\end{theorem}

\begin{proof}
   See Appendix\,B. 
\end{proof}

\begin{remark}
 Under Assumption\,\ref{assumption1} and further regularity conditions on the payoff function (i.e., differentiability), the PDE \eqref{eq:pde} can be understood in the classical sense (see e.g., \cite[Theorem IV.4.1]{Fleming06}). This means that the PDE condition can be imposed by the technique of PINN using automatic differentiation of neural networks. 
\end{remark}

\subsection{Physics-informed RL (PIRL) Framework} \label{sec:framework}

Here we present the proposed PIRL framework. % by combining the problem statement in Proposition\,\ref{prop:RL} and the PDE characterization of the safety probability in Theorem\,\ref{thm:hjb}. 
While in principle any RL algorithms can be considered, here we focus on an extension of %combination of 
the Deep Q-Network (DQN) algorithm \cite{Mnih2013} as a simple but practical example. 
The optimal action-value function $q^\ast(s,a)$ will be estimated  by using a function approximator $Q(s, a; \theta)$ with the parameter $\theta$. 

The proposed algorithm %combining DQN and PINN 
is presented in Algorithm\,\ref{alg:DQN+PINN}. %, where the difference with the DQN algorithm is shown by red.  
The overall structure %of the algorithm
follows from the DQN algorithm, while we added new statements in the lines 14 to 19 to take samples for PINN and modified the loss function $L$ used in the line 21.
Following the framework of PINN~\cite{Raissi2019}, our loss function $L$ consists of the three terms of $L_\mathrm{D}$ for the data loss of the original DQN, $L_\mathrm{P}$ for the physics model given by the PDE \eqref{eq:pde}, and $L_\mathrm{B}$ for the boundary conditions \eqref{eq:boundary_tau_zero} and \eqref{eq:boundary_unsafe}, \emph{i.e.}, % for which we provide detailed explanations in subsequent paragraphs. 
\begin{align}
    L = & L_\mathrm{D} + \lambda L_\mathrm{P}
    + \mu L_\mathrm{B},   
    \label{eq:loss_DQN+PINN} 
\end{align}
where $\lambda$ and $\mu$ are the weighting coefficients, and the specific form of each loss is given below.  %The sets $\mathcal{S}_\mathrm{D}$, $\mathcal{Y}_\mathrm{D}$, $\mathcal{S}_\mathrm{P}$, $\mathcal{A}_\mathrm{P}$, $\mathcal{S}_\mathrm{B}$ and $\mathcal{A}_\mathrm{B}$ of data used for these loss terms are given later.  
After the initializations of the replay memory $\mathcal{D}$, the function approximator $Q$, and its target function $\hat{Q}$, % \cite{Mnih2013,mnih15},  
%
%
% In the DQN algorithm, the technique of experience of replay   %\cite{Lin1992} 
% is used, where the agent’s data is stored in a replay memory $\mathcal{D}$, and the data is randomly sampled from different time steps for the purpose of reducing non-stationarity and correlation between updates \cite{Mnih2013}. 
% The line 1 states the initialization of this replay memory. 
% Also for enhancing the stability of learning process, besides the function approximator $Q$, a separate network $\hat{Q}$ is used for generating the target values of the function $Q$ \cite{mnih15}. 
% After the initialization of these networks in the lines 2 and 3, 
the main loop starting at the line 4 iterates $M$ episodes, and the inner loop starting at the line 6 iterates the time steps of each episode. 
Each episode starts with the initialization of the state $s_0 = [{h}_0, {x}_0^\top]^\top$ in the line 5, which is sampled from the distribution $P_\mathrm{D}$ given by 
\begin{align}
    P_\mathrm{D}(s_0) = 
    \begin{cases}
       1/|\Omega_\mathrm{D}|, &  {h}_0 = \tau_\mathrm{D} \land {x}_0 \in \Omega_\mathrm{D}, \\
       0, & \text{otherwise}, 
    \end{cases}
\end{align}
where $\tau_\mathrm{D} \in \R_+$ is the time interval of the data acquired through the DQN algorithm, which can be smaller than $\tau$. %equal to the outlook horizon $\tau$ of interest or can be $\tau_\mathrm{D} < \tau$ when long-term samples are costly. 
The set $\Omega_\mathrm{D} \subset \mathbb{X} $ is the domain of possible initial states, and $|\Omega_\mathrm{D}|$ is its volume. 
%Thus, the length of each episode scales with the parameter $T_\mathrm{D}$, and it can be chosen as equal to the outlook horizon $T$ or shorter when using PINN allows generalization to unseen region. % as explained later. 
At each time step $k$, through the lines 7 to 10, a sample of the transition $(s_k, a_k, r_k, s_{k}^\prime )$ of the state $s_k$, the action $a_k$, the reward $r_k$, and the next state $s_k^\prime$ is stored in the replay memory $\mathcal{D}$.  
In the lines 11 to 13, a random minibatch $\mathcal{S}_\mathrm{D}$ of transitions is taken from $\mathcal{D}$, and the set $\mathcal{Y}_\mathrm{D}$ of the target values is calculated using the target q-function $\hat{Q}$, where the $j$-th element $y_j$ of $\mathcal{Y}_\mathrm{D}$ is given by\footnote{The index $j$ is independent of the time step $k$. Random sampling from different time steps improves the stability of learning process by reducing non-stationarity and correlation between updates \cite{Mnih2013}.} 
\begin{align}
    y_j =  
    \begin{cases}
      r_j, & \textrm{for terminal}~s_{j}^\prime, \\
      r_j + \max_{a} \hat{Q}(s_j^\prime, a; \hat{\theta}), & \textrm{otherwise}. 
        \end{cases}  \label{eq:Q_target}
\end{align}
%The full algorithm of DQN is presented in Algorithm\,\ref{alg:DQN}. 
%It stores $N$ experience tuples of transition $(s_k, a_k, r_k, s_{k}^\prime )$ of the state $s_k$, the action $a_k$, the reward $r_k$, and the next state $s_k^\prime$. 
%During the inner loop of the algorithm, a minibatch update is applied using samples from the replay memory.  
%By denoting the sampled set of experiences by $\mathcal{S}_\mathrm{D}$, 
Then, the loss function $L_\mathrm{D}$ 
%of the original DQN algorithm 
is given by 
\begin{align}
    & L_\mathrm{D}(\theta; \mathcal{S}_\mathrm{D}, \mathcal{Y}_\mathrm{D}) = \dfrac{1}{|\mathcal{S}_\mathrm{D}|} \sum_j  (y_j - Q(s_j, a_j; \theta))^2.
    \label{eq:DQNloss}
\end{align}
% where $\mathcal{Y}_\mathrm{D}$ is the set of the target of the Q function for each experience, and its element $y_j$ is given by 
% \begin{align}
%     y_j =  
%     \begin{cases}
%       r_j, & \textrm{for terminal}~s_{j}^\prime, \\
%       r_j + \max_{a} \hat{Q}(s_j^\prime, a; \hat{\theta}), & \textrm{otherwise}, 
%         \end{cases}  \label{eq:Q_target}
% \end{align}
% where $\hat{Q}$ is a separate network for generating the targets $y_j$ aimed at improving the stability, of which weight $\hat\theta$ is updated as $\hat{\theta} = \tau \theta + (1-\tau) \hat{\theta}$ with a smoothing factor $\tau$. 
%
%
%
To calculate the loss term $L_\mathrm{P}$, a random minibatch $\mathcal{S}_\mathrm{P} =  \{ s_l \}$ is taken at the line 15. 
Each element $s_l=[h_l, x_l^\top]^\top$ is sampled from the distribution $P_\mathrm{P}$ given by 
\begin{align}
    P_\mathrm{P}(s_l) = 
    \begin{cases}
       1/(\tau |\Omega_\mathrm{P}|), & {h}_l \in [0, \tau] \land {x}_l \in \Omega_\mathrm{P}, \\
       0, & \text{otherwise}, 
    \end{cases}
\end{align}
with $\Omega_\mathrm{P} \subset \SSet$ that specifies the domain where the PDE \eqref{eq:pde} is imposed. 
In the line 16, the set of greedy actions $\mathcal{A}_\mathrm{P} = \{ a_l^\ast \}$ is calculated by $a_l^\ast = \arg \max_a Q(s_l, a; \theta)$. 
Then, the PDE loss $L_\mathrm{P}$ can be defined as 
\begin{align}
 L_\mathrm{P}(\theta; \mathcal{S}_\mathrm{P}, \mathcal{A}_\mathrm{P}) =  \dfrac{1}{|\mathcal{S}_\mathrm{P}|} \sum_l W_\mathrm{P}(s_l, a_l^\ast; \theta)^2 
\end{align}
with the residual function $W_\mathrm{P}(s_l, a_l^\ast; \theta)$ given by 
\begin{align} 
   W_\mathrm{P}(s_l, a_l^\ast; \theta) := 
   &  %\dfrac{ \partial Q(s_l,a_l^\ast; \theta)}{ \partial s }
   \partial_s Q(s_l,a_l^\ast; \theta) 
   \tilde{f}(s_l, a_l^\ast) \notag \\
   &  + \dfrac{1}{2} \mathrm{tr} \left[ \tilde{\sigma}(s_l, a_l^\ast)\tilde{\sigma}(s_l, a_l^\ast)^\top %\dfrac{ \partial^2 Q(s_l,a_l^\ast;\theta)}{ \partial s^2 } 
   \partial_s^2 Q(s_l,a_l^\ast;\theta) 
   \right]. 
  \end{align} 
% with the function $\tilde{f}$ and $\tilde{\sigma}$ defined as 
% \begin{align}
%     \tilde{f}(s, a) := 
%     \begin{bmatrix}
%        -1 \\   f(x, a)
%     \end{bmatrix},
%     \quad
%     \tilde{\sigma}(s) := 
%     \begin{bmatrix}
%         0 \\ \sigma(x)        
%     \end{bmatrix}. 
% \end{align}
%
%%
%$\mathcal{S}_\mathrm{P} = \{ s_l \}$ is a set of states sampled from a distribution $P_\mathrm{P}(s)$, which takes a non-zero value inside the boundary, and the $\mathcal{A}_\mathrm{P} = \{ a_l^\ast \}$ is the set of greedy action $a_l^\ast = \arg \max_a Q(s_l, a; \theta)$.
%
%combine PINN, we consider two additional loss terms to the original DQN data loss $L_\mathrm{D}$: 
For the boundary loss $L_\mathrm{B}$, as stated in the line 18, a minibatch $\mathcal{S}_\mathrm{B} = \{ s_m \}$ with $s_m = [h_m, x_m^\top]^\top$ is taken by using the %a set of states randomly sampled based on a 
distribution $P_\mathrm{B}(s)$ given by %that takes a non-zero value only on the boundary. 
\begin{align}
    P_\mathrm{B}(s) = 
    \begin{cases}
       1/(2|\Omega_\mathrm{P}|), & {h_m}=0 \land {x_m} \in \Omega_\mathrm{P}, \\
       1/(2\tau |\Omega_\mathrm{B}|), & h_m \in [0, \tau ] \land {x} \in \Omega_\mathrm{B},  \\
       0, & \text{otherwise}.
    \end{cases}
\end{align}
where %$\Omega_\mathrm{B} := \partial \Omega \cap \partial \SSet$. 
$\Omega_\mathrm{B} \subset \partial\SSet $ stands for the lateral boundary. 
%where $\Omega_\mathrm{B}$ stands for the subset of $\mathbb{X}$ where the condition \eqref{eq:boundary_tau_zero} is imposed, and  $T_\mathrm{B}$ and $\mathbb{X}_\mathrm{B}$ for the time interval and the subset of $\partial \mathcal{C}$ that the condition \eqref{eq:boundary_unsafe} is imposed, respectively. 
The loss $L_\mathrm{B}$ can be defined as 
% enforces the boundary conditions \eqref{eq:boundary_tau_zero} and \eqref{eq:boundary_unsafe}, and given by 
\begin{align}
 L_B(\theta; \mathcal{S}_\mathrm{B}, \mathcal{A}_\mathrm{B}) =  \dfrac{1}{|\mathcal{S}_\mathrm{B}|} \sum_m W_\mathrm{B}(s_m, a_m^\ast; \theta)^2,  
\end{align}
with the set $\mathcal{A}_\mathrm{B} = \{ a_m^\ast \}$ %, the set 
of greedy action % $a_m^\ast = \arg \max_a Q(s_m, a; \theta)$ for each state $s_m$,
%where $\mathcal{S}_\mathrm{B} = \{ s_m \}$ is a set of states randomly sampled based on a distribution $P_\mathrm{B}(s)$ that takes a non-zero value only on the boundary. 
%By using $\mathcal{A}_\mathrm{B} = \{ a_m^\ast \}$, the set of greedy action $a_m^\ast = \arg \max_a Q(s_m, a; \theta)$ for each state $s_m$, 
and the residual $W_\mathrm{B}$ given by 
\begin{align}
   W_\mathrm{B}(s_m, a_m^\ast; \theta) = 
   Q(s_m, a_m^\ast; \theta) - l_\epsilon(x_m). 
   % \begin{cases}
   %     Q(s_m, a_m^\ast; \theta) - \mI [x \in \SSet ], & \tau = 0 , \\
   %     Q(s_m, a_m^\ast; \theta) &  \phi(x) \le 0.
   % \end{cases}
\end{align}
Finally, at the line 21, the parameter $\theta$  is updated to minimize the total loss $L$ based on a gradient descent step.  
The parameter $\hat{\theta}$ of the target function $\hat{Q}$ used in \eqref{eq:Q_target} is updated at the line 22 with a smoothing factor $\eta \in (0,1]$.

%%%%%%%%%%%%%%%%%%%%%%%%%%%%%%%%%%%
\begin{algorithm}[t]
\caption{DQN integrated with PINN} \label{alg:DQN+PINN} 
\begin{algorithmic}[1]
\State Initialize replay memory $\mathcal{D}$ to capacity $N_\mathrm{mem}$
\State Initialize function $Q$ with random weights $\theta$
\State Initialize target Q function $\hat{Q}$ with weights $\hat{\theta} = \theta$
\For{ episode $ = 1: M$}
 %   \State Initialize time step $k \gets 0$
    \State Initialize state $s_0 = [h_0, x_0^\top]^\top\sim P_\mathrm{D}(s_0)$
    \For{ $k= 1:N_\mathrm{f} $}
        \State /*  Emulation of experience  */
        \State  With probability $\epsilon$ select a random action $a_k$ 
        \Statex \hspace{9.5mm} otherwise select $a_k = \arg\max_a Q(s_k, a; \theta)$
        \State Execute action $a_k$
        \Statex \hspace{9.5mm} and observe reward $r_k$ and the next state $s_{k+1}$
        \State Store transition $(s_k, a_k, r_k ,s_k^\prime )$ in $\mathcal{D}$
        \State /*  Sample experiences  */
        \State Sample random minibatch $\mathcal{S}_{D}$ of 
        \Statex \hspace{9.5mm} transitions $(s_j, a_j, r_j ,s_{j}^\prime )$ from $\mathcal{D}$
        \State Set  $\mathcal{Y}_\mathrm{D} = \{ y_j \}$ as in Eq.\,\eqref{eq:Q_target}
        \State /*  Sample minibatch for PDE */
        \State Sample minibatch $\mathcal{S}_{P}$ with $s_l \sim P_\mathrm{P}(s)$ 
        \State Set $\mathcal{A}_\mathrm{P}$  with $a_l^\ast = \arg \max_{a} Q(s_l,a;\theta)$
        \State /*  Sample minibatch for boundary conditions */
        \State Sample minibatch $\mathcal{S}_{B}$ with $s_m \sim P_\mathrm{B}(s)$ 
        \State Set $\mathcal{A}_\mathrm{B}$ with $a_m^\ast = \arg \max_{a} Q(s_m,a;\theta)$
        \State /*  Update of weights */
        \State Perform a gradient descent step on
        \Statex \hspace{9.5mm} loss function $L$ in Eq.\, \eqref{eq:loss_DQN+PINN} with respect to $\theta$
        \State Update target weight $\hat{\theta} \gets \eta \theta + (1-\eta ) \hat{\theta}$
%        \State Increment $k \gets k+1$
%    \Until{Condition $s_{k-1} \in \mathcal{S}_\mathrm{abs}$ holds}
        \EndFor
\EndFor
\end{algorithmic}
\end{algorithm}

With this algorithm, the length of each episode scales with the parameter $\tau_\mathrm{D}$, and it can be chosen as equal to the outlook horizon $\tau$ or shorter. 
When we set $\tau_\mathrm{D} < \tau$, 
the PDE constraint is imposed on the entire time domain of $[0, \tau]$, and the safety probability is learned only from experiences with shorter time interval $\tau_\mathrm{D}$. 
In this case, the safety probability predicted by the PINN has bounded error %, and the error scales linearly with the training loss 
\cite[Theorem\,6]{Wang2023}. 
This is beneficial in the situation where long-term trajectories for rare unsafe events can be hardly obtained.

\section{Numerical Example} \label{sec:example}

This section demonstrates the effectiveness of the PIRL algorithm through a proof-of-concept numerical example.  
Consider the SDE \eqref{eq:sde} with the state space $\mathbb{X} := \{ x = [x_1, x_2]^\top |\, x \in \R^2 \}$, the control space $\mathbb{U} := [-1, 1] \subset \R $, %the 2-dimensional Brownian motion $W_t \in \R^2$, 
and the functions $f$ and $\sigma$ given by 
\begin{align}
    f(x, u) = 
    \begin{bmatrix}
      -x_1^3 - x_2 \\  x_1 + x_2 + u   
    \end{bmatrix}, \quad
    \sigma(x, u) = 
    \begin{bmatrix}  0.2 & 0\\ 0 & 0.2 \end{bmatrix}. 
\end{align}
This example is based on \cite{Beard1997} and has an unstable equilibrium point $x^\ast = [0, 0]^\top$, satisfying $f(x^\ast, 0) = 0$. %, and 
% A nominal stabilizing controller can be designed by using 
% the technique of feedback linearization and %to have an auxiliary input $v$ with $u(x) = 3x_1^5 + 3 x_1^2 x_2 - x_2 + v$, and 
% then applying the LQR theory \cite{Beard1997}. % with $Q = \mathrm{diag}(1,1)$ and $R=1$. 
Here, we consider the safe set $\SSet$ given as follows\footnote{To satisfy Assumption\,\ref{assumption1}(d), 
one can arbitrarily chose a sufficiently large bounded set $\mathcal{C}^\mathrm{c}$ to cover a part of unsafe region in $\mathbb{X}$ of interest.}:
\begin{align}
    \SSet = \{ x \in \mathbb{X} \,|\, (1- x_2^2) > 0 \}. 
\end{align}
%\Figref{fig:safety_prob}\,(a) shows the safety probability under the nominal controller for the outlook horizon $\tau=2.0$. % and $2.0$, respectively. 
%The red arrows in the figure show the vector field of the deterministic part of the dynamics.  
%The safety probability in these figures were obtained by a standard Monte Carlo simulation. 
%
%
%\Figref{fig:safety_prob}\,(b) shows the %result of 
%safety probability given by  the proposed PIRL algorithm. 
%The proposed algorithm was applied to this example. 
%In the learning process, 
For the implementation of the proposed DQN based algorithm, which admits a discrete action space, the control was restricted to $u \in \{ -1.0,\, -0.5, \, 0, \, 0.5, \, 1.0 \}$. % to implement the proposed DQN algorithm. %, which admits a discrete action space. 
This restriction does not affect the results when the underlying optimal control problem has a  “bang-bang” nature \cite{Rubies-Royo2019}. %, and indeed, it has been confirmed that the learned controller takes $-1$ or $1$ near the boundary of the safe set. 
%The proposed algorithm was implemented by modifying \texttt{DQN agent} provided in MATLAB R2023a \cite{dqn_agent}. 
For the function approximator $Q$, we used a neural network with 3 hidden layers with 32 units per layer and the hyperbolic tangent (\texttt{tanh}) as the activation function. 
%The activation function was chosen as hyperbolic tangent function (\texttt{tanh}). 
The batch sizes are $|S_\mathrm{D}|= |S_\mathrm{B}|=|S_\mathrm{P}|=64$, and the weighting coefficients were chosen as $\mu = 1$ and $\lambda = 1\times10^{-2}$. 
The initial state of each episode %during the training 
was randomly sampled from $P_\mathrm{D}$ with $\Omega_\mathrm{D}  = \{ [x_1, x_2] \in \R^2 \,|\, |x_1|\le 1.5, |x_2|\le 1.0 \}$.  
The set $\Omega_\mathrm{P}$ and $\Omega_\mathrm{B}$ were given as 
$\Omega_\mathrm{P} = \{ [x_1, x_2]  \,|\, |x_1|\le 1.5, |x_2|\le 0.9 \}$ and $\Omega_\mathrm{B} = \{ [x_1, x_2]  \,|\, |x_1|\le 1.5, |x_2|= 1.0 \}$, respectively. 
Our implementation is available at here\footnote{ {\url{https://github.com/hoshino06/PIRL_ACC2024}} }.

\begin{figure}[!t]
 \centering
\subcaptionbox{Nominal controller}
   { \includegraphics[width=0.465\linewidth]{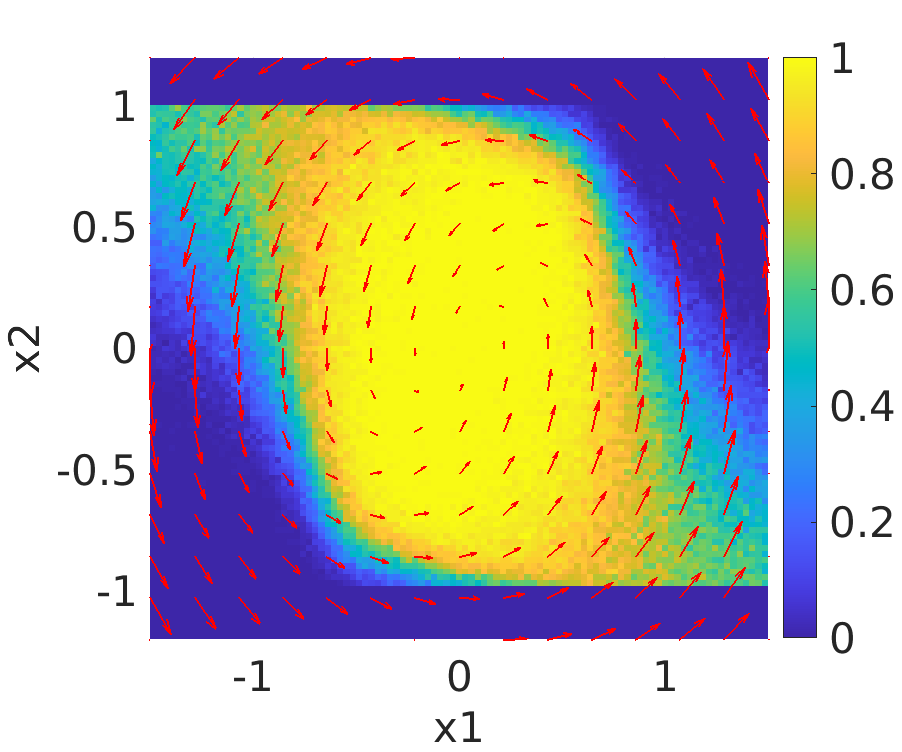} }   
\subcaptionbox{Proposed PIRL}
   { \includegraphics[width=0.465\linewidth]{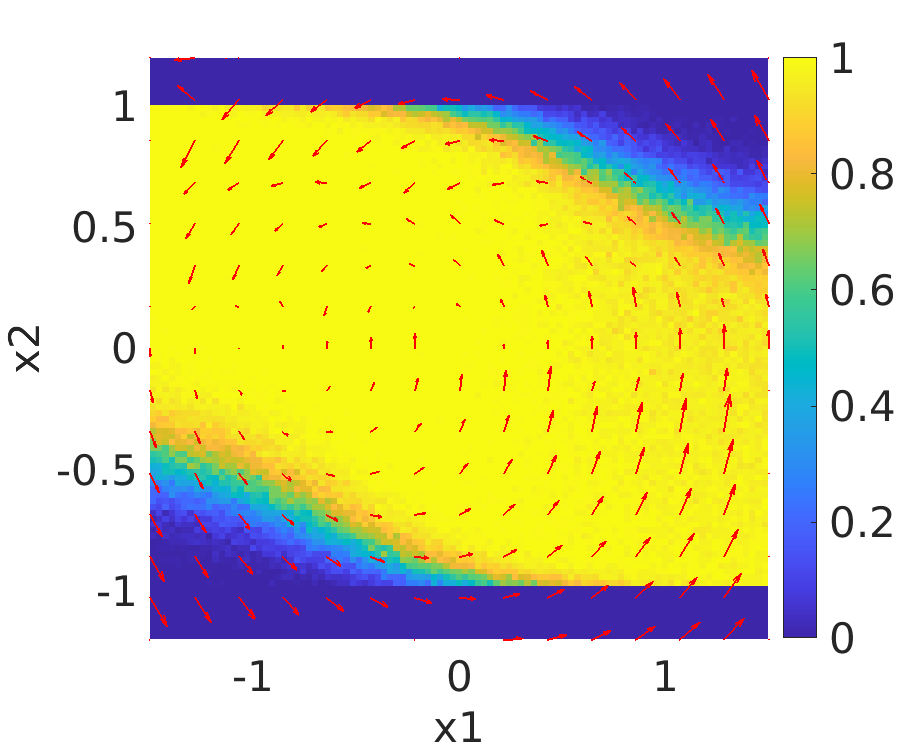} }  
 \caption{Safety probability for the outlook horizon of $\tau=2.0$.}  \label{fig:safety_prob}
 \vspace{-3mm}
\end{figure}

\begin{figure}[!t]
 \centering
 \includegraphics[width=0.95\linewidth]{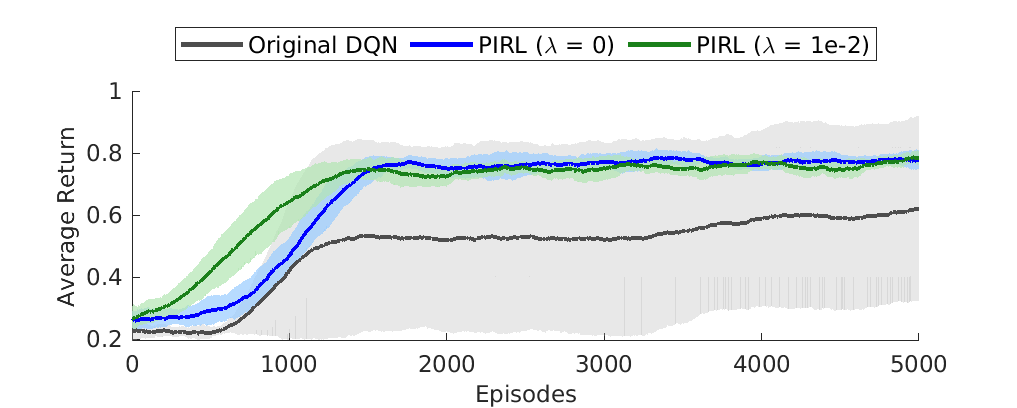} 
 \caption{Learning progress.}  \label{fig:learning_process}
  \vspace{-3mm}
\end{figure}

\textit{Benefit of maximally safe policy and probability.}
\Figref{fig:safety_prob} shows the safety probability for the outlook horizon $\tau=2.0$ with (a) nominal controller and (b) controller learned by PIRL. 
Here, the nominal controller was obtained by using 
the technique of feedback linearization and then applying the LQR theory as in \cite{Beard1997}. 
The safety probabilities in the figure were calculated by a standard Monte Carlo simulation (the estimation accuracy by the function approximator $Q$ will be discussed later). 
The red arrows in the figure show the vector field of the deterministic part of the dynamics.  
%By using PIRL, the area with high safety probability has been expanded due to exploration of maximally safe controller, implying that the resultant system can move to wider states due to the expansion of feasible region of the state space. 
When the learned policy is used, the system achieves higher safety probability.  
Thus, when the learned probability is used in probabilistic safety certificates such as \cite{Wang2022}, it allows the system to explore wider regions.

\begin{comment}
\begin{figure}[!t]
\subcaptionbox{Safety probability vs episodes \label{fig:prob_ve_episodes}}{
 \includegraphics[width=1.0\linewidth]{data/fig2a_learning_process.eps}} 
\subcaptionbox{Comparison with reward shaping \label{fig:reward_shaping}}{
 \includegraphics[width=0.9\linewidth]{data/fig2b_reward_shaping.eps}}
 \caption{Safety probability during the learning process.}  \label{fig:learning_process}
\end{figure}
\end{comment}

\textit{Efficient learning despite sparse reward.}
\Figref{fig:learning_process} shows the training progress. 
Besides the plot for the proposed PIRL shown by \emph{green}, the \emph{black} line shows the result with the original DQN, and the  \emph{blue} line the PIRL with $\lambda= 0$, which means only boundary conditions were imposed. 
The solid curves correspond to the mean of eight repeated experiments, and the shaded region shows their standard deviation. 
With the original DQN, the agent has to learn only from sparse zero/one rewards, and often fails to find a safe policy. 
In contrast, with the proposed PIRL, a safe policy can be found despite the sparse rewards, % especially at the initial phase,
%because imposing physics loss allows information of the reward at $\tau=0$ to propagate to the region of $\tau > 0$. % without exploration by DQN. 
and the averaged return (corresponds to the safety probability) rises with fewer samples especially at the initial phase. 

One of the most common solutions to the issue of sparse reward is reward shaping~\cite{Nilaksh2024}. 
For example, one could design a reward $r_\mathrm{rs}$ to include information about the distance from the boundary of the safe set: 
\begin{align}
    r_\mathrm{rs}(s) := r(s) + c (1 - x_2^2). 
\end{align}
\Figref{fig:reward_shaping} shows the comparison of the averaged safety probability achieved at the initial phase of the training, where $c=0.05$. % for the results of reward shaping. 
It can be seen that the proposed PIRL can learn with fewer experiences as well as the reward shaping. % (implying with fewer unsafe events). 
This is because imposing physics loss allows propagation of reward information from neighbors and boundaries, and can serve as an alternative to reward shaping.

% Riedmiller, Learning by Playing Solving Sparse Reward Tasks from Scratch, https://proceedings.mlr.press/v80/riedmiller18a.html
% Abhishek Ranjan, Barrier Functions Inspired Reward Shaping for Reinforcement Learning, https://arxiv.org/abs/2403.01410

\begin{figure}[!t]
 \centering
 \includegraphics[width=0.9\linewidth]{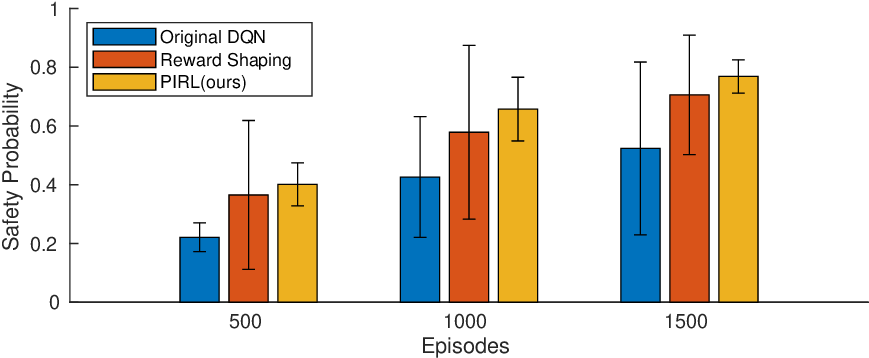}
 \caption{Comparison with reward shaping.} \label{fig:reward_shaping}
  \vspace{-2mm}
\end{figure}

\textit{Generalization of PIRL.}
%Finally, we discuss the benefits of generalization by PIRL. 
\Figref{fig:estimation_accuracy} shows the accuracy of the safety probability for $\tau = 2.0$ estimated by the function approximater $Q$ learned with $\tau_\mathrm{D} \le \tau$, with and without PDE condition\footnote{Without PDE condition means $\lambda=0$ but imposing the boundary conditions. With PDE condition,  
too large $\lambda$ led to an unstable behavior in the training process due to an excessive exploitation of the greedy policy. Here it was chosen as $5\times 10^{-3}$ for $\tau_\mathrm{D} = 1.5$ and $3\times 10^{-3}$ for $\tau_\mathrm{D} = 1.0$ to avoid an unstable behavior in the training process.}. 
The bar plot shows the mean squared error between the output of the learned neural network and the Monte Carlo calculation over 
$10 \times 10$ equally distributed points in the state space, and error bar represents its standard deviation over eight repeated experiments.  
On the other hand, \figref{fig:num_unsafe_events} shows the number of unsafe events during the training process with and without the PDE constraint. 
By reducing $\tau_\mathrm{D}$, the number of unsafe events can be reduced, but there is a stringent trade-offs between the estimation accuracy and the length of $\tau_\mathrm{D}$ when the PDE constraint is not imposed ($\lambda=0$). 
In contrast, with the proposed PIRL, the safety probability is accurately estimated without acquiring data no longer than $\tau_\mathrm{D}$ while reducing the number of unsafe events. 
This is beneficial especially in situations when safety must be ensured for a \textit{longer} period than sampled trajectories and when safe actions must be learned \textit{without} sufficiently many rare events and unsafe samples.  

\begin{figure}[!t]
 \centering 
 \subcaptionbox{Estimation accuracy for $\tau = 2.0$ \label{fig:estimation_accuracy}}{
 \includegraphics[width=0.85\linewidth]{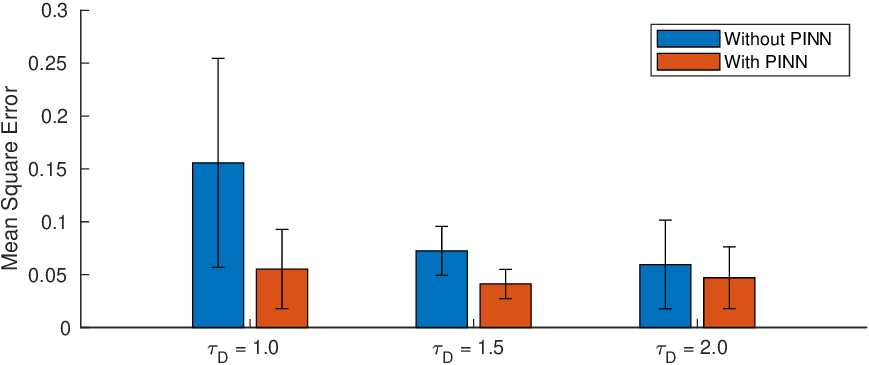}}
 \subcaptionbox{Number of unsafe events \label{fig:num_unsafe_events}}{
 \includegraphics[width=0.85\linewidth]{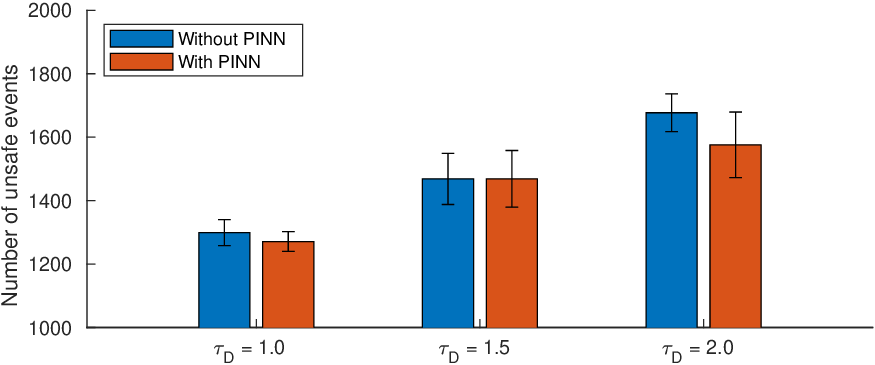}}
 \caption{Results of generalization with $\tau_\mathrm{D} \le \tau$.}  \label{fig:generalization}
 \vspace{-2mm}
\end{figure}

%Specifically, the initial state $s_0$ is sampled such that $T_\mathrm{D} = 1.0$ or $1.5$ at each case, but we test the estimation result of the safety probability for the outlook horizon of $T=2.0$. 
%The figure shows the mean square error of the estimated safety probabilities of $100 \times 100$ equally distributed points in the state space of $\{ (x_1, x_2) \in \mathbb{X} \,| \, |x_1|<1.5, |x_2|<1.2 \}$ compared with the result shown in \figref{fig:safety_prob}(b) with $T_\mathrm{D}=2.0$. 
%By imposing $L_\mathrm{P}$, the safety probability can be better estimated for $T=2.0$ without acquiring data through experiences due to generalization to unseen regions by physics model, and is beneficial in the situation where long-term trajectories for rare unsafe events can be hardly obtained.  

%\figref{fig:generalization}(a) shows the result with $\lambda=0$, and it is hard to estimate the safety probability for the outlook horizon longer than the experiences without any prior knowledge (the figure looks different from \figref{fig:safety_prob}(b)). 
%In contrast, as show in \figref{fig:generalization}(b), imposing $L_\mathrm{P}$ enables inference of the safety probability for $T=2.0$ due to the generalization to unseen regions by combining data from experiences with the physics model.  

%%%%%%%%%%%%%%%%%%%%%%%%%%%%%%%%%%%%%%%%%
\section{Conclusions}

In this paper, we proposed a Physics-informed Reinforcement Learning (PIRL) for efficiently estimating the safety probability under maximally safe actions. %while exploring a safe strategy to expand the feasible regions of the state space. 
%ddressed the problem of safety probability estimation under a maximally safe controller. 
This was based on the exact characterization of the safety probability as the value function of an RL problem and the derivation of a PDE condition satisfied by the action-value function. 
%, we developed a physics-informed Reinforcement Learning (RL) framework where a maximally safe controller is explored while training a Physics-Informed Neural Network (PINN) for accurate estimation of safety probability.  
The effectiveness of PIRL has been demonstrated through an example based on the Deep  Q-Network algorithm integrated with the technique of  Physics-informed Neural Network (PINN). % and its application to a planer dynamical system. 
Future work includes application of this framework to estimate safety probability in real-world tasks such as autonomous driving, and its use in \emph{e.g.}, safe RL.


%%%%%%%%%%%%%%%%%%%%%%%%%%%%%%%%%%%%%%%%%%%%%%%%%%%%%%%%%%%%%%%%%%%%%%%%%%%%%%%%

%%%%%%%%%%%%%%%%%%%%%%%%%%%%%%%%%%%%%%%%%%%%%%%%%%%%%%%%%%%%%%%%%%%%%%%%%%%%%%%%



% \subsection{Proof of Proposition\,\ref{prop:RL}} \label{proof:prop_RL}

%%%%%%%%%%%%%%%%%%%%%%%%%%%%%%%%%%%%%%%%%%%%%%%%%%%%%%%%%%
\section*{APPENDIX}

%%%%%%%%%%%%%%%%%%%%%%
\subsection{Proof of Proposition\,\ref{prop:RL}} \label{proof:prop_RL}

\begin{proof}
From the fact that $\tilde{H}_k = H_k$ and $\tilde{X}_k = X_k$ hold for $S_k \notin \mathcal{S}_\mathrm{abs}$, 
%inspired by the formulation of reach-avoid problems developed in \cite{Summers2010}, 
the safety probability $\SProb^u(\tau, x)$ can be rewritten as follows: 
\begin{align}
   \SProb^u(\tau, x)
   %%%%%%%%%%%%%%%%%%
    = &\mathbb{E}\left[ P_{N(\tau)}  \,\middle|\,  S_0 = s, u \right]
    \label{eq:sp_rewrite}
\end{align}
with the symbol $P_{k}$ is defined as 
\begin{align}
    P_{k} :=  \prod_{j=0}^k \mathds{1}[ \tilde{X}_j \in \SSet ]. 
\end{align}
Then, this can be further transformed into a form of sum of multiplicative cost as follows:
\begin{align}    
\SProb^u(\tau, x)
    %%%%%%%%%%%%%%%%%%
    = &\mathbb{E}\bigl[ \left( P_{N(\tau)} \right) \mathds{1}[ \tilde{H}_N \in \mathcal{G} ]  \,\big|\, S_0 = s, u \bigr]  \notag \\
   %%%%%%%%%%%%%%%%%%
    = &\mathbb{E}\left[ \sum_{k=0}^{N(\tau)} \left( P_{k}\right) \mathds{1}[ \tilde{H}_k \in \mathcal{G}  ] \,\middle|\, S_0 = s, u \right].  %\notag \\
    \label{eq:sp_final}
\end{align}
%where the expectations above are conditioned on the control action $A_k=a_k$ for $k\in \mZ_+$. 
Here, %introducing $S_\mathrm{abs}$ does not affect the above, since 
 the transformations from \eqref{eq:sp_rewrite} to \eqref{eq:sp_final} is based on the fact that $\mI[ \tilde{H}_k \in \mathcal{G}]$ is $1$ if $k=N(\tau)$ and $0$ otherwise. 
Given this form of representation, the proof will be completed by showing that the expectation of the return $G = \sum_{k=0}^\infty r(S_k)$ is equal to \eqref{eq:sp_final}. 
First, consider the case where $\tilde{X}_k$ stays inside the safe set $\SSet$ for all $k=0, \dots, N(\tau)$, \emph{i.e.}, the trajectory is safe.  
In this case, we have 
\begin{align}
  & P_{k} = 1, \quad \mI[S_k \notin \mathcal{S}_\mathrm{abs} ] = 1,  \quad \forall k \in \{ 0,\, \dots, \, N(\tau) \}.
\end{align}
Since $\mI[\tilde{H}_k \in \mathcal{G}] = 0$ for all $k\ge N(\tau)+1$, we have 
\begin{align}
   G = & \sum_{k=0}^\infty  \mI[\tilde{H}_k \in \mathcal{G}]\,\mI[S_k \notin \mathcal{S}_\mathrm{abs} ] \notag\\
   = & \sum_{k=0}^{N(\tau)} \mathds{1}[ \tilde{H}_k \in \mathcal{G} ] \,\mI[S_k \notin \mathcal{S}_\mathrm{abs} ] \notag \\
   = &\sum_{k=0}^{N(\tau)} \left( P_{k} \right) \mathds{1}[ \tilde{H}_k \in \mathcal{G} ]. \label{eq:return_safe}
\end{align}
Next, consider the case where $\tilde{X}_{\bar{k}} \notin \SSet$ for some $\bar{k} \in \{ 0, \dots, N(\tau) \}$, \emph{i.e.}, the trajectory is unsafe. 
Then, we have
\begin{align}
  & P_{k} = \mI[S_k \notin \mathcal{S}_\mathrm{abs} ] = 1, \quad \forall k \in \{ 0,\, \dots, \,\bar{k}-1 \}, \\
  & P_{k} = \mI[S_k \notin \mathcal{S}_\mathrm{abs} ] = 0,  \quad \forall k \in \{ \bar{k},\, \dots, \,N(\tau) \}. 
\end{align}
Thus, the return becomes
\begin{align}
   G 
   %= \sum_{k=0}^{N_\mathrm{f}} R(S_k) = \sum_{k=0}^{K}  R(S_k) 
   = & \sum_{k=0}^{N(\tau)} \mathds{1}[ \tilde{H}_k \in \mathcal{G} ] \,\mI[S_k \notin \mathcal{S}_\mathrm{abs} ] \notag \\
   = & \sum_{k=0}^{\bar{k}-1} \mathds{1}[ \tilde{\tau}_k \in \mathcal{G} ] + \sum_{k=\bar{k}}^{N(\tau)} 0 \cdot \mathds{1}[ \tilde{H}_k \in \mathcal{G} ] \notag \\
   = & \sum_{k=0}^{N(\tau)} \left( P_{k} \right) \mathds{1}[ \tilde{H}_k \in \mathcal{G} ]. \label{eq:return_unsafe}
\end{align}
Thus, the expectation of the return $G$ over all possible trajectories, which is represented either by  \eqref{eq:return_safe} or \eqref{eq:return_unsafe}, is equivalent to the safety probability $\SProb^u$ given in \eqref{eq:sp_final}. 
Also, the function $v^u$ takes a value in $[0,1]$, since the return $G$ takes $0$ or $1$. 
\end{proof}

%%%%%%%%%%%%%%%%%%%%%
\subsection{Proof of Theorem\,\ref{thm:hjb}} \label{proof:thm_hjb}

\begin{proof}
First, consider the following function for the SDE \eqref{eq:sde} with the exit-time $T_\mathrm{e} = \min(T_{\SSet^\mathrm{c}}, t_\mathrm{f})$: 
\begin{align}
 V^{\bar{u}}_\epsilon(t_\mathrm{s},x) := \mE [l_\epsilon (X_{T_\mathrm{e}}^{t_\mathrm{s},x; \bar{u}} )]. 
\end{align}
Then, it follows from \cite[Theorem 4.7]{MohajerinEsfahani2016} that under Assumption\,\ref{assumption1}(a)-(d), the function $V^\ast_\epsilon(t_\mathrm{s},x) := \sup_{\bar{u} \in \mathcal{U}_{t_\mathrm{s}} }  V^{\bar{u}}_\epsilon(t_\mathrm{s},x)$ is a viscosity solution of the following PDE: 
\begin{align} \label{eq:reach-avoid_HJB}
 \sup_{ a \in \mathbb{U}} \mathcal{L}^a  V_\epsilon^\ast(t_\mathrm{s},x) = 0, & \quad \forall (t_\mathrm{s}, x) \in [0, t_\mathrm{f}) \times \SSet,
\end{align}
where $\mathcal{L}^a$ is the Dynkin operator defined as 
\begin{align}
  \mathcal{L}^a \Phi(t,x) := & \partial_t \Phi(t,x) + f(x,a)^\top \partial_x \Phi(t,x) 
  \notag \\
  & \quad + \dfrac{1}{2}\text{tr}[\sigma(x,a) \sigma(x,a)^\top \partial_x^2 \Phi(t,x) ],
\end{align}
and the boundary condition given by 
\begin{align}
 V_\epsilon^\ast(t_\mathrm{s},x) = l_\epsilon(x), \quad \forall (t, x)  \in [0, t_\mathrm{f}] \times \SSet^\mathrm{c} \cup \{t_\mathrm{f} \} \times \R^n. 
\end{align}
The continuity of the function $V_\epsilon^\ast$ follows from Lipschitz continutity of the payoff function $l_\epsilon$ and uniform continuity of the stopped solution process \cite[Proposition 4.8]{MohajerinEsfahani2016}. 

Here, %By taking $A_\epsilon = \{ x \in \SSet \,|\, \text{dist}(x, \SSet^\mathrm{c}) \ge \epsilon \}$ and $B = \SSet^\mathrm{c}$, 
the function $V^{\bar{u}}_\epsilon(t_\mathrm{s}, x)$ can be rewritten as 
\begin{align}
 V^{\bar{u}}_\epsilon(t_\mathrm{s},x) & = \mE[\mI[X_{T_\mathrm{e}}^{t_s,x;\bar{u}}\in \SSet]\, l_\epsilon (X_{T_\mathrm{e}}^{t,x; \bar{u}} ) ]
  \notag \\
  & = \mE[\mI[X_{t}^{t_s,x;\bar{u}}\in \SSet, \, \forall t \in [t_\mathrm{s}, t_\mathrm{f}]] \, l_\epsilon (X_{T_\mathrm{e}}^{t,x; \bar{u}} )],  \notag
\end{align}
where the above transformations follows from the fact that $l_\epsilon(X_{T_\mathrm{e}}^{t_\mathrm{s},x; \bar{u}}) \neq 0$ only if $X_{T_\mathrm{e}}^{t_\mathrm{s},x;\bar{u}}\in \SSet$.
Then, with the function $\Psi^{u}_\epsilon(\tau, x)$ given by 
\begin{align}
  & \Psi^{u}_\epsilon(\tau, x) := \mathbb{E} \left[ \prod_{k=0}^{ N(\tau) } \mI[X_k \in \SSet ]\, l_\epsilon(X_{N(\tau)}) ~\middle | ~ X_0= x, u  \right],  \notag
\end{align}
from the continuity of the function $l_\epsilon$ and Assumption\,\ref{assumption1}(e), we have 
$V^u_\epsilon(t_\mathrm{f}-\tau, x) = \lim_{\Delta t \to 0}\Psi^{u}_\epsilon(\tau, x)$.
Furthermore, with the same arguments as in the proof of Proposition\,\ref{prop:RL}, we have $ \SProb^{u}_\epsilon(\tau, x) = v^{u}_\epsilon(s) $
with 
\begin{align}
  &v^{u}_\epsilon(s) := \mathbb{E} \left[ \sum_{k=0}^{N(\tau)} r_\epsilon( S_k ) ~\middle | ~ S_0= s, u  \right]. 
\end{align}
Since we have 
\begin{align}
     v^{u}_\epsilon(s) = \max_{a \in \mathbb{U}} q^\ast_\epsilon(s, a) \xrightarrow[\Delta t \to 0]{} V^\ast_\epsilon(t_\mathrm{f} -\tau, x), 
\end{align}
and $V^\ast_\epsilon(t_\mathrm{s}, x)$ satisfies the PDE \eqref{eq:reach-avoid_HJB}, the function $q^\ast_\epsilon(s, a)$ satisfies the following PDE as $\Delta t \to 0$: 
\begin{align} \label{eq:HJB_q_func_pre}
 & \sup_{a\in \mathcal{U}} \partial_s q^\ast_\epsilon(s, a^\ast) \tilde{f}(s, a) 
      \notag \\ & \hspace{10mm} 
      +\dfrac{1}{2} \mathrm{tr} \left[ \tilde{\sigma}(s, a)\tilde{\sigma}(s, a)^\top \partial_s^2 q^\ast_\epsilon(s, a^\ast)  \right] = 0.    
\end{align}
Here, from $q^\ast_\epsilon(s, a) = r_\epsilon(s) + \mE[ v^\ast_\epsilon(s^\prime) | s, a] $, where $s^\prime$ is the next state given the current state $s$ and the input $a$, we have 
\begin{align}
    a^\ast & = \argsup_{a \in \mathbb{U} } \, r(s) + \mE[ v^\ast_\epsilon(s^\prime) | s, a] \notag \\ 
    & = \argsup_{a \in \mathbb{U} } \mE[ v^\ast_\epsilon(s^\prime) | s, a] \notag \\ 
    & = \argsup_{a \in \mathbb{U} } \dfrac{\mE[ v^\ast_\epsilon(s^\prime) | s, a] -v^\ast_\epsilon(s)}{\Delta t}
\end{align}
where the above transformation is based on the fact that $r_\epsilon(s)$ and $v^\ast_\epsilon$ are independent of $a$. 
Thus, from the Ito's Lemma, $a^\ast$ maximizes the right-hand side of \eqref{eq:HJB_q_func_pre} as $\Delta t \to 0$, and substituting $a^\ast$ gives \eqref{eq:pde}:
\begin{align}
  & \partial_s  q^\ast_\epsilon(s, a^\ast) 
      \tilde{f}(s, a^\ast) 
      \notag \\ & \hspace{5mm} 
      +\dfrac{1}{2} \mathrm{tr} \left[ \tilde{\sigma}(s,a^\ast)\tilde{\sigma}(s,a^\ast)^\top  \partial_s^2 q^\ast_\epsilon(s, a^\ast)  %\bigg|_{s=s, a=a^\ast}  
      \right] = 0. \notag
\end{align} 
\end{proof}

\section*{ACKNOWLEDGMENT}
The authors would thank Maitham F. AL-Sunni and Haoming Jing for critical reading of drafts and their helpful comments.

\bibliographystyle{IEEEtran}
\bibliography{reference}

% \begin{thebibliography}{99}

% \bibitem{c1} G. O. Young, ÒSynthetic structure of industrial plastics (Book style with paper title and editor),Ó 	in Plastics, 2nd ed. vol. 3, J. Peters, Ed.  New York: McGraw-Hill, 1964, pp. 15--64.
% \end{thebibliography}

\end{document}